\documentclass[a4paper,UKenglish]{lipics}

\usepackage[ruled,vlined]{algorithm2e}
\usepackage{algorithmic}

\newtheorem{observation}{Observation}

\title{SONIK: Efficient In-situ All Item Rank Generation using Bit Operations}

\author{Sourav Dutta}
\affil{Max Planck Institute for Informatics,\\
  Saarbr\"ucken, Germany.\\
  \texttt{sdutta@mpi-inf.mpg.de}}
\authorrunning{S. Dutta} 


\begin{document}

\maketitle

\begin{abstract}
Sorting, a classical combinatorial process, forms the bedrock of numerous algorithms with varied 
applications. A related problem involves efficiently finding the corresponding ranks of all the 
elements -- catering to rank queries, data partitioning and allocation, etc. Although, the element 
ranks can be subsequently obtained by initially sorting the elements, such procedures involve 
$O(n \log n)$ computations and might not be suitable with large input sizes for hard real-time 
systems or for applications with data re-ordering constraints.

This paper proposes $SONIK$, a non-comparison linear time and space algorithm using bit operations 
inspired by {\em radix sort} for computing the ranks of all input integer elements, thereby providing 
implicit sorting. The element ranks are generated {\em in-situ}, i.e., directly at the corresponding 
element position without re-ordering or recourse to any other sorting mechanism.
\end{abstract}

\section{Introduction and Related Work}
\label{sec:intro}

\textbf{Preliminaries:} Given a collection of $n$ integer elements, $\sigma = \{e_1, e_2, \cdots, e_n\}$, 
{\em sorting} refers to the procedure of arranging the elements in accordance with a definitive ordering, 
e.g., numeric or lexicographic. Mathematically, sorting produces a permutation of the input items such 
that the elements in the final list are ordered as $e_i \leq e_j, \forall i \leq j$.

Efficient sorting algorithms provide a classical area of research in the domain of algorithms and beyond. 
The literature provides an enormous wealth of applications involving sorting, as a major operational step, 
such as comparison and searching, statistical analysis, and priority-based scheduling to name a few. Advanced 
algorithms such as knapsack problem~\cite{kp,num}, minimum spanning tree~\cite{mst,fib,krus}, and network 
analysis also employ sorting as an intermediate stage. Several sorting algorithms such as bubblesort~\cite{knuth,cormen}, 
quicksort~\cite{quick}, mergesort~\cite{knuth}, heapsort~\cite{heap}, etc., based on different strategies such 
as {\em divide-and-conquer}, {\em exchange}, {\em partitioning}, and {\em greedy approach}~\cite{cormen} have 
been proposed to cater to different applications and scenarios. 

Sorting algorithms can be classified into two categories -- {\em comparison sort} involving element-wise 
comparison or {\em counting based sort} using the count of items. Comparison sort such as mergesort, 
heapsort, etc., have been shown to be theoretically bounded by $O(n \log n)$~\cite{cormen}, where $n$ represents 
the number of elements. On the other hand, counting-based sorts are bounded by $O(n)$ assuming certain 
characteristics of the input data such as uniform distribution (in bucket sort~\cite{cormen,buck}) or a small 
distribution range of the input element (in counting sort~\cite{count}). Further, implementational optimizations 
have also been studied for improving the practical efficiency of common sorting algorithms by using intelligent data 
structures, such as multi-key quicksort~\cite{multi} and trie-based radix sort.

Sorting algorithms are also characterized by certain behavioral properties, such as {\em in-place} (involving 
the use of $O(1)$ extra storage space), {\em stable} (preserving the appearance order of duplicate elements in 
the sorted order)~\cite{cormen}, and {\em adaptive} (taking advantage of pre-existing ordering in the input)~\cite{adapt}. 
The crucial nature of sorting techniques in various domains has also led to the development of several parallelized 
and disk-aware sorting procedures such as bitonic sort~\cite{bitonic}, odd-even sort~\cite{oe}, etc.

\textbf{Motivation:} An interesting and closely related problem involves the computation of the {\em ranks} of all 
the elements in an input list, where the rank of an item is defined as its position in the final sorted list. For 
example, given an input collection $\sigma = \{7, 2, 1, 5, 4\}$ the corresponding rank list for the elements is 
$\rho = \{5, 2, 1, 4, 3\}$, based on the non-decreasing key sorted order $\{1, 2, 4, 5, 7\}$ of $\sigma$. In fact, several 
applications involving scheduling, order-statistics, job priority, data partitioning, etc.,~\cite{cormen,knuth} 
use element ranks as the working metric. For example, consider the scenario of task scheduling from a job queue 
within a server pool with heterogeneous configurations, and the rank of a job to denote its priority. As such, 
the scheduling algorithm needs to compute the rank of all the jobs in the queue for the scheduler to 
appropriately allocate the job with the highest priority (i.e., highest rank) to the server with the greatest 
compute power and/or resource availability, and so on. Further, certain restrictions as to the re-ordering of 
items might apply in certain application scenarios. Hence, efficient approaches for computation of all item 
ranks to facilitate enhanced downstream processing provide an interesting area of study for such scenarios. 
Unfortunately, limited work exists in this regard, to the best of our knowledge.

\textbf{Application:} Consider an input list of $n$ objects, where each object is represented by a $d$ dimensional 
vector of integers corresponding to the values of its $d$ dimensions. Applications such as {\em multi-criteria searching} 
require the input objects to be ordered such that most (or all) of the $d$-dimensional values are nearly (or completely) 
sorted. As such, the objects need to be independently sorted based on their $i^{th}$ dimension to generate $d$ such ordered 
lists. Finally, the $d$ ordered lists are to be suitably combined into a final sorted list of the $n$ objects such that 
most of the dimensions are nearly sorted. Observe that, an expensive join operation is required to be performed across 
the individual sorted lists (for each of the dimensions) as the different lists might have different ordering of the 
objects (based on the dimensional values). However, an efficient approach finding the ranks of all the objects for the 
individual dimensions in a pre-defined or their input appearance order would eliminate the need for the expensive merge 
procedure, thereby significantly improving the performance of the application. In this paper, we aim to tackle this 
{\em all-item ranking} problem by proposing an effective algorithm to compute a rank list containing the ranks of the 
input elements in the order of their appearance, i.e., in-situ.

\textbf{Problem Statement:} Formally, {\em {\texttt given an input list of $n$ elements, $\sigma = \{e_1, e_2, \cdots, e_n\}$, 
the problem of {\em all-item ranking} entails the efficient computation of the rank list $\rho = \{r_1, r_2, \cdots, r_n\}$, 
where $r_i$ denotes the {\em rank} of element $e_i, i \in [1,n]$, with rank $r_i$ referring to the position of element $e_i$ 
in the sorted order of $\sigma$. We further aim to obtain the element ranks possibly without item re-ordering or 
use of sorting techniques which might be restricted in certain application scenarios (e.g., distributed job scheduling).}}

\textbf{State-of-the-art:} Na\"ively, we could initially sort the input list and thereafter obtain the ranks of individual 
elements by searching its position within the ordered list. Observe that, sorting of the input list can be 
performed in $O(n)$ best case by use of non-comparison based sorting approaches. However, such methods assume certain 
characteristics of the input elements (such as small range, uniform distribution, etc., as discussed earlier). Moreover, 
the computation of the rank list for the input elements (in the order of appearance in the input list) involves subsequent 
searching of each of the elements in the sorted order for their positions (i.e., ranks). For example, in our previous example, 
for obtaining the rank of $e_1 = 7~ (\in \sigma)$ requires the searching of $e_1$ in the sorted list $\{1,2,4,5,7\}$. 
This involves additional $O(\log n)$ computations for each element, using binary search procedure. Hence, the total complexity 
for obtaining the rank list is bounded by $O(n \log n)$. 

Certain scenarios might restrict the modification of the original order of elements, and hence re-ordering of elements for in-place 
sort may be infeasible (e.g., re-ordering of jobs and data allocated to servers might be forbidden or have significant network 
communication overheads). However, an additional $O(n)$ space to create a copy of the data for sorting is acceptable in 
most cases, as the generation of the rank list would also require $O(n)$ extra space.

The {\em quickselect}~\cite{qs} algorithm working on the {\em divide-and-conquer} principle provides the state-of-the-art 
$O(n)$ approach to obtain the item with rank $r$ from an unsorted list. However, finding the ranks of all elements degenerates 
its performance to quadratic complexity, i.e., $O(n^2)$, for finding the ranks of all items, i.e., $\forall r \in [1,n]$. 

To alleviate such problems, the augmented {\em order-statistics tree}~\cite{cormen} was proposed to efficiently store and index the 
input elements in a tree-based structure by storing extra statistics at each node of the tree (e.g., number of descendants 
in the left subtree, etc.). Subsequent querying for the item having rank $r$ can then be performed in $O(h)$, where $h$ 
represents the height of the tree structure. The use of height-balanced tree data structures like {\em red-black tree}~\cite{rbt} 
or {\em AVL tree}~\cite{avl} provides an upper bound of $O(\log n)$ for $h$ and and for finding the item with rank $r$. Thus, the 
construction and representation of input elements using an order-statistics tree and rank computation for all the elements 
(i.e., rank list) therein can be performed in $O(n \log n)$ time. However, such approaches suffer from involved tree rotational 
procedures for height balancing and also incur huge space requirements due to extra information and pointer storage within the nodes 
of the data structure. 

Hence, we observe that state-of-the-art computation of the ranks of all input elements can be performed in $O(n \log n)$.

\textbf{Contribution:} In this paper, we propose the $SONIK$ algorithm for the {\em all-item ranking} problem, i.e., finding the 
rank list for an input list of integer elements in their input appearance order, in \textbf{$O(nk)$} time, where $n$ is the number 
of input elements and $k$ denotes the number of bits required to represent the elements. Considering, $k$ to be constant (in the 
order of a few tens, since 32/64 bits are sufficient to encode huge numbers) we obtain a {\em linear} time algorithm, and show 
that the space requirement is also linear in $n$. 

$SONIK$ employs bit comparison operations (similar to the working principle of LSB radix-sort~\cite{radix}) to provide an efficient 
and practical approach for real-time scenarios, and also an {\em anytime algorithm}~\cite{any} obtaining the partial ranks of elements 
based on the $i$ least-significant bits after $i$ iterations. We also show that no re-ordering of elements or dependence on other 
sorting approaches (as in radix-sort) is required by our approach leading to {\em in-situ} generation of the element ranks. Further, 
the generated ranks of the elements provide an inherent sorting of the input list in linear time.

\textbf{Outline:} The remainder of the paper is organized as follows: Section~\ref{sec:sonik} describes the 
detailed working of the proposed $SONIK$ algorithm. Section~\ref{sec:opt} discusses a few implementational 
optimizations along with avenues for generalization to handle different input scenarios and enforce sorting 
properties such as stability. Finally, Section~\ref{sec:conc} concludes the paper, followed by literary 
references in this problem domain.

\section{\texttt{SONIK} Algorithm}
\label{sec:sonik}

In this section, we discuss the detailed working principle of our proposed {\em SONIK} algorithm. Without loss of generality, we assume 
the input list of elements, $\sigma$, to consist of non-negative integers without the presence of any duplicates, and that the item ranks 
are to be computed based on the non-decreasing sorted order of $\sigma$. Generalizations for handling other input characteristics and scenarios 
will be later presented in Section~\ref{sec:opt}. 

As a running example, consider $\sigma = \{7,2,1,5,4\}$ to be the input list and $\rho = \{r_1, \cdots, r_5\}$ as the final element rank 
list to be computed, where $r_i$ represents the rank of element $e_i \in \sigma$. Further, let $k$ be the number of bits required to encode 
the elements in $\sigma$; hence $k=3$ for our above example.

The rank list, $\rho$ is at first initialized to $1$ and $SONIK$ operates on the bit representations of the input elements, with processing 
proceeding from the least-significant bit (LSB) to the most-significant bit (MSB), i.e., from the rightmost to the leftmost bit. Thus, bit 
encoding of the input items, $\sigma_{bit} = \{111, 010, 001, 101, 100\}$ and initial $\rho = \{1,1,1,1,1\}$ forms the starting configuration 
for our algorithm. Starting from the LSB, $SONIK$ iterates over the $k$ bits of the input elements, and at each iteration performs two major 
computational phases -- (1) {\em Rank Updation} and (2) {\em Rank Consolidation}. 

At the onset of iteration $i$, the $i^{th}$ least significant bits of all the elements in $\sigma$ are extracted and based on the associated 
bit values, two lists, $L_0^i$ and $L_1^i$, are created to store the current ranks of elements with their corresponding bits set to $0$ and 
$1$ respectively. 

\begin{observation}
\label{ob:rank}
Given only the least significant $i$ bits of non-negative integers, elements with the $i^{th}$ bit set (i.e., elements corresponding to the ranks in list 
$L_1^i$) have a higher numerical value compared to the elements associated with list $L_0^i$ (i.e., having $i^{th}$ bit value as $0$), and hence should have 
larger ranks (based on non-decreasing sorted order of $\sigma$).
\end{observation}

For our example, during the first iteration (i.e., considering only the least significant bits) elements $\{7,1,5\} \in \sigma$ have the 
rightmost bit set to $1$ (from $\sigma_{bit}$) and hence their corresponding ranks are used to obtain list $L_1^1$ as $\{1,1,1\}$. Similarly, 
corresponding to the remaining elements, we obtain $L_0^1 = \{1,1\}$ as shown in Figure~\ref{fig:example}. Based on the above observation, 
considering only the LSB, elements $\{7,1,5\}$ (of list $L_1^1$) should be assigned larger ranks as compared to elements $\{2,4\}$ (in list $L_0^1$).

$SONIK$ then performs the {\em rank updation} and {\em rank consolidation} phases separately for the elements in lists $L_0^i$ and $L_1^i$, as 
discussed next.

\subsection{$L_0$ List Processing}
\label{ssec:l0}

Initially, the ranks of the elements corresponding to the list $L_0^i$ (obtained above) are appropriately processed and updated by the two stages of 
$SONIK$ as follows:

\noindent \textbf{Rank Updation Phase:} Using the above computed lists, $L_0^i$ and $L_1^i$, at each iteration the {\em rank updation} phase computes 
the following three operating variables:

\begin{itemize}
	\item ${\bf \pi_m^i}$ :-- denotes the {\em minimum} rank $m$ among the elements having the $i^{th}$ bit position set to $1$, i.e., 
minimum value present in list $L_1^i$;
	\item ${\bf \pi_{\leq m}^i}$ :-- depicts the {\em maximum} rank {\em lesser than or equal to} $\pi_m^i$ among the elements corresponding 
to list $L_0^i$, i.e., elements with $i^{th}$ bit set to $0$. The value of $\pi_{\leq m}^i$ is considered to be $0$ if no such rank value is found; and, 
	\item ${\bf \pi_{\geq m}^i}$ :-- represents the {\em minimum} rank {\em greater than or equal to} $\pi_m^i$ present in list $L_0^i$, i.e., 
among the ranks of elements with $i^{th}$ bit set to $0$. Alternatively, this value is set to $\pi_{\leq m}^i$ is no such rank is found.
\end{itemize}

From Observation~\ref{ob:rank}, based on the $i^{th}$ bit, the elements corresponding to list $L_1^i$ are numerically greater than the elements 
pertaining to list $L_0^i$, and hence their ranks need to be appropriately updated (increased). As such, the elements associated to list $L_0^i$ 
having a rank greater than the minimum rank among the elements of $L_1^i$ ($\pi_m^i$), are processed by $SONIK$ during the rank updation stage. 
On the other hand, the ranks of elements of $L_0^i$ with current ranks less than or equal to $\pi_m^i$ remain unchanged during this step, as 
shown for the first element in Figures~\ref{fig:update}(a) and (b).

\begin{figure}[t]
\centering
	\includegraphics[width=\columnwidth,height=1.5in]{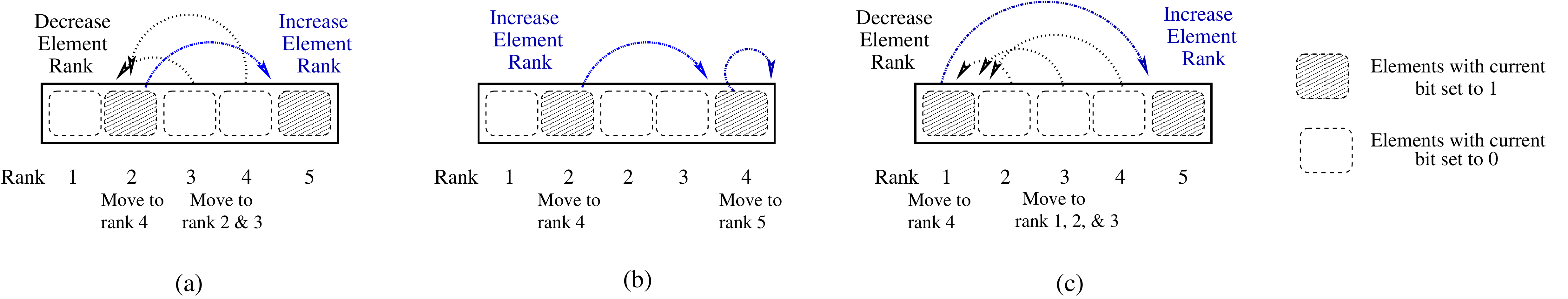}
	\caption{Working of the {\em Rank Updation} procedure in $SONIK$ for various scenarios.}
	\label{fig:update}
\end{figure}

Observe that, in general the element(s) in $L_0^i$ with rank $\pi_{\geq m}^i$ should now be re-assigned a rank of $\pi_m^i$, as it represents 
the currently numerically smallest element that was previously assigned a higher rank than $\pi_m^i$ (in list $L_1^i$) during the iterations based 
on $i-1$ LSB. Hence, in Figure~\ref{fig:update}(a), the item with rank $3$ should now be assigned a rank of $2$.

However, no such comparative deductions can be made for elements in $L_0^i$ that have the same rank value of $\pi_m^i$ (i.e., when $\pi^i_{\leq m} 
= \pi_m^i$). In these cases, the rank of such element(s) is kept unchanged during the current iteration and its final rank is computed by subsequent 
iteration steps. Hence, in unison, items in $L_0^i$ with ranks less than or equal to $\pi_m^i$ undergo no rank re-assignment as they are appropriately 
represented by their current ranks. In Figure~\ref{fig:update}(b), elements having rank $2$ and $3$ (depicted with unfilled boxes) represent 
such scenarios with unchanged current rankings.

Hence, for each element corresponding to list $L_0^i$ having a rank greater than $\pi_m^i$, $SONIK$ decrements the current rank by the {\em rank 
update} parameter value, $\delta_0$, defined as follows:
\begin{align}
\delta_0 &= \left\{
	\begin{array}{ll}
	0 & \mbox{  if  } \pi_m^i = \pi_{\leq m}^i \\
	\pi_{\geq m}^i - \pi_m^i & \mbox{  otherwise}
	\end{array}
\right. \\
\text{Thus,} &\text{ the updated element ranks are obtained by,} \nonumber \\ 
\rho_i &= \left\{
	\begin{array}{ll}
	\rho_i & \mbox{  if  } \rho_i \leq \pi_m^i \\
	\rho_i - \delta_0 & \mbox{  otherwise}
	\end{array}
\right. \qquad \qquad \text{[$\forall i, e_i \in L_0^i$]}
\end{align}

The pseudo-code for the {\em rank updation} phase is presented in Algorithm~\ref{algo:ru}, while Figure~\ref{fig:update} provides a pictorial description 
of the rank re-assignment strategy under different scenarios. The updated ranks of the elements corresponding to $L_0^i$ are then provided as input to the 
{\em rank consolidation} phase, wherein the ranks are consolidated to represent contiguous positions.

\begin{algorithm}[t]
\caption{{\em Rank Updation Procedure for $L_0^i$}, $RUL_0()$}
\label{algo:ru}
{\small{
	\begin{algorithmic}[1]
	\REQUIRE Input Element List ($\sigma$), Rank List ($\rho$), Element set $S_0$ corresponding to list $L_0^i$, and Operating variables 
	($\pi_m^i$, $\pi_{\leq m}^i$, and $\pi_{\geq m}^i$)
	\ENSURE Re-assignment of ranks to elements corresponding to list $L_0^i$.

	\medskip

	\STATE $A \gets \pi_m^i$, $B \gets \pi_{\geq m}^i$, and $C \gets \pi_{\leq m}^i$

	\IF{A = C}
		\STATE $\delta_0 \gets 0$
	\ELSE
		\STATE $\delta_0 \gets \left(B - A\right)$
	\ENDIF

	\FOR{each element $e_i \in \sigma$ and $e_i \in S_0$}
		\IF{$\rho_i > A$}
			\STATE $\rho_i \gets \rho_i - \delta_0$
		\ENDIF
	\ENDFOR

	\STATE Return $\rho$
	\end{algorithmic}
}}
\end{algorithm}

\noindent \textbf{Rank Consolidation Phase:} Given the updated element ranks as obtained above, consolidation stage simply processed the ranks and to 
make them contiguous, so as to eliminate ``holes'' in the element rankings obtained. For example, an input rank list $\tau = \{1, 2, 4, 6\}$ is 
consolidated to $\tau'= \{1,2,3,4\}$. We now make the following observation to enable the efficient performance of the {\em consolidation} phase for the 
$SONIK$ algorithm.

\begin{observation}
\label{ob:sort}
The updated rank list for the elements pertaining to list $L_0^i$ is in non-decreasing sorted order.
\end{observation}

Since the ranks of all the candidate items (in $L_0^i$) are reduced by a common factor, the {\em rank update} parameter, during consolidation, 
the original ordering among these elements are inherently preserved. Further, since the re-assigned ranks are greater than $\pi_m^i$, and element ranks 
less than $\pi_m^i$ are not modified, the sorted property is preserved for the updated rank list (as the original rank list $\rho$ was initially sorted).

We now show that with the use of an extra storage space and two array position pointers, the {\em rank consolidation} phase can be performed in linear time 
in the input rank list size. An additional array $E$ of size $2*n$ (described later in Observation~\ref{ob:max}) is created, and each element $t_i \in \tau$ 
is mapped to the array position $E[t_i]$. The array position pointer, $p_1$ provides the next free slot, while the other position pointer $p_2$ refers to the 
next rank to be consolidated. The rank value of $E[p_2]$ is then accordingly set at position $E[p_1]$ based on the value at $E[p_1-1]$. This procedure is 
repeated until all the input ranks are consolidated as consecutive values.

For our example list $\tau$, we create and initialize $E[8] = \{1,2,-,4,-,6\}$ with the pointers initially set as $p_1 = p_2 = 1$. Iterating over 
the structure $E$, we obtain $p_1 = 3$ (the first free position of $E$) and $p_2 = 4$ (position of the next rank after $p_1$). The array $E$ is then 
updated as $E[p_1] = E[p_1-1]+1$ and $E[p_2] = -$. Hence, $E$ is modified to $\{1,2,3,-,-,6\}$. Finally, the first $|\tau|$ elements of 
$E=\{1,2,3,4,-,-\}$ are returned as the consolidated rank list for $\tau$. Algorithm~\ref{algo:rc} provides the pseudo-code for the rank consolidation 
procedure. Observe that, the sorted order of the input rank list enables the rank consolidation procedure to operate with a single pass of the array $E$, 
thereby running in linear time in the size of $\tau$.

\begin{algorithm}[t]
\caption{{\em Rank Consolidation Procedure}, $RCP()$}
\label{algo:rc}
{\small{
	\begin{algorithmic}[1]
	\REQUIRE Updated rank list $\tau$ of elements in list $L^i_j$ from the rank updation procedure
	\ENSURE Consolidated consecutive ranking for elements in $\tau$

	\medskip

	\STATE $n \gets |\tau|$, $p_1 \gets 1$, and $p_2 \gets 1$
	\STATE Create array $E$ of size $2*n$

	\FOR{element $t_i$ in $\tau$}
		\STATE $E[t_i] \gets t_i$
	\ENDFOR

	\WHILE{$p_1 \leq n \land p_2 \leq n$}
		\STATE $p_1 \gets $ next free position in $E$
		\STATE $p_2 \gets $ next rank value after position $p_1$

		\IF{$p_1 = 1$}
			\STATE $E[p_1] = E[p_2]$
		\ELSE
			\STATE $E[p_1] = E[p_1-1]+1$
		\ENDIF
		\STATE Remove entry $E[p_2]$
	\ENDWHILE

	\STATE $max \gets E[n]$
	\STATE Return first $n$ elements of $E$ as the consolidated rank list of $\tau$ along with $max$
	\end{algorithmic}
}}
\end{algorithm}

\subsection{$L_1$ List Processing} 

After the processing of elements of list $L_0^i$, the ranks of the corresponding items are appropriately updated and consolidated to reflect 
their current ranks based on the $i$ least significant bits. The ranks of the remaining elements corresponding to list $L_1^i$ are now adjusted 
for assignment of higher ranks (as they are numerically greater) based on $max$, the maximum rank assigned to an element during the above rank 
consolidation procedure on list $L_0^i$. Observe that, currently the minimum ranked element in $L_1^i$ (having a rank of $\pi_m^i$) should now 
be assigned a rank of $max+1$, and hence the current ranks of all elements pertaining to list $L_1^i$ are incremented by the update parameter, 
$\delta_1 = max + 1 - \pi_m^i$. Hence, $\rho_i = \rho_i + \delta_1$ $(\forall i, e_i \in L_1^i)$ forms the {\em rank updation} step for 
elements of list $L_1^i$ as shown in Algorithm~\ref{algo:ru1}.

The obtained output ranks for list $L_1^i$ from the above rank update procedure are then provided to the {\em rank consolidation} step for 
generating contiguous rankings, similar to the procedure described in Section~\ref{ssec:l0}.

\begin{algorithm}[t]
\caption{{\em Rank Updation Procedure for $L_1^i$}, $RUL_1()$}
\label{algo:ru1}
{\small{
	\begin{algorithmic}[1]
	\REQUIRE Input Element List ($\sigma$), Rank List ($\rho$), Element set $S_1$ corresponding to list $L_1^i$, $max$ returned by consolidation 
	phase on list $L_0^i$, and Operating variable $\pi_m^i$
	\ENSURE Re-assignment of ranks to elements corresponding to list $L_1^i$.

	\medskip

	\STATE $A \gets \pi_m^i$
	\STATE $\delta_1 \gets \left(max + 1 - \pi_m^i\right)$

	\FOR{each element $e_i \in \sigma$ and $e_i \in S_1$}
		\STATE $\rho_i \gets \rho_i + \delta_1$
	\ENDFOR

	\STATE Return $\rho$
	\end{algorithmic}
}}
\end{algorithm}

We now reason about the $2*n$ size bound of the extra storage array during the rank consolidation phase.

\begin{observation}
\label{ob:max}
The maximum rank obtained by an element during the rank updation stage is bounded by $2*n$.
\end{observation}

\begin{proof}
Consider the current minimum rank $\pi_m^i$ of an element in $L_1^i$ to be $1$, i.e., numerically it denotes the minimum element based on 
its $i-1$ least significant bits. Observe that, at the start of the $i^{th}$ iteration, the previous rank updation and consolidation steps 
produce an ordered contiguous rank list ranging from $1$ to $n$, the number of input elements in $\sigma$. Assume list $L_1^i$ to also contain 
the element with the maximum current rank $n$. Considering all other elements of $\sigma$ to be in list $L_0^i$ (as shown in 
Figure~\ref{fig:update}(c)) and the list to have been processed as described in Section~\ref{ssec:l0}, the maximum rank value, $max$ returned 
by the consolidation procedure on $L_0^i$ can be $n-2$ (as list $L_1^i$ is assumed to contain two elements).

In our current setting, the rank update parameter, $\delta_1 = max + 1 - \pi_m^i = n-2 + 1 - 1 = n-2$. The rank updation procedure on list 
$L_1^i$ might thus generate a maximum of $n + \delta_1 = 2*n-2$ as the updated rank for the maximum ranked element in $L_1^i$. Hence, the 
maximum rank assigned is upper bounded by $2*n$, leading to the bound on the size of the extra storage array $E$ for the rank consolidation 
phase.

Further, rank updation for list $L_0^i$ decrements the current element ranks by the $\delta_0$ parameter value. As the current maximum rank of 
any element can be $n$, the maximum rank obtained after the updation phase of $SONIK$ remains bounded by $2*n$.
\end{proof}

\noindent \textbf{Overall Computation:} The element rank list, $\rho_i$, obtained from the above computations (i.e., the rank updation and rank 
consolidation stages) by $SONIK$ over the two lists $L_0^i$ and $L_1^i$, depicts the intermediate ranking output at the end of the $i^{th}$ 
iteration. The rank list, $\rho_i$ is then fed as an input for processing during the next iteration step. This procedure is repeated until all 
the $k$ bits of the input elements in $\sigma$ have been processed, and the final rank list, $\rho$ generated after the $k^{th}$ 
iteration provides the output for our algorithm, containing the final ranks of all the input elements in the order of their appearance 
in $\sigma$. The pseudo-code for the complete working of the proposed $SONIK$ algorithm is described in Algorithm~\ref{algo:sonik}. 
A working C++ code for $SONIK$ can be obtained from \url{https://www.dropbox.com/s/zhn4q1hnho61mg6/SONIK.cpp?dl=0}.

\begin{algorithm}[t]
\caption{{\em $SONIK$ Rank Generation Algorithm}}
\label{algo:sonik}
{\small{
	\begin{algorithmic}[1]
	\REQUIRE Input element list $\sigma$
	\ENSURE Element Rank list $\rho$ capturing the ranks of elements in $\sigma$

	\medskip

	\STATE $n \gets |\sigma|$
	\STATE Create rank list $\rho$ of size $n$ and initialize all entries to $1$
	\STATE $k \gets $ number of bits required to encode elements of $\sigma$

	\FOR{$i:=1$ to $k$}
		\STATE Extract the $i^{th}$ bit, $b_i^j$ of element $e_j$ of $\sigma$, for $j \in [1,n]$
		\STATE Create empty lists $L_0^i$ and $L_1^i$
		\IF{$b_i^j = 0$}
			\STATE Add current rank of element $e_j$, $\rho_j^i$ to list $L_0^i$
		\ELSE
			\STATE Add current rank of element $e_j$, $\rho_j^i$ to list $L_1^i$
		\ENDIF

		\STATE Compute operating variable $\pi_m^i$, $\pi_{\leq m}^i$, and $\pi_{\geq m}^i$ from lists $L_0^i$ and $L_1^i$

		\medskip

		\STATE {\em RANK UPDATION:} Perform $RUL0()$ on list $L_0^i$ with $\sigma$, $\rho^i$, and operating variables
		\STATE {\em RANK CONSOLIDATION:} Call procedure $RCP()$ for list $L_0^i$ with the above output

		\medskip
		\medskip

		\STATE {\em RANK UPDATION:} Execute $RUL1()$ on list $L_1^i$ with $\sigma$, $\rho^i$, and operating variables
		\STATE {\em RANK CONSOLIDATION:} Call procedure $RCP()$ for list $L_1^i$ with the above output
	\ENDFOR

	\STATE Output $\rho$ as the final rank list of elements of $\sigma$
	\end{algorithmic}
}}
\end{algorithm}

\begin{figure}[!ht]
\centering
	\includegraphics[width=\columnwidth]{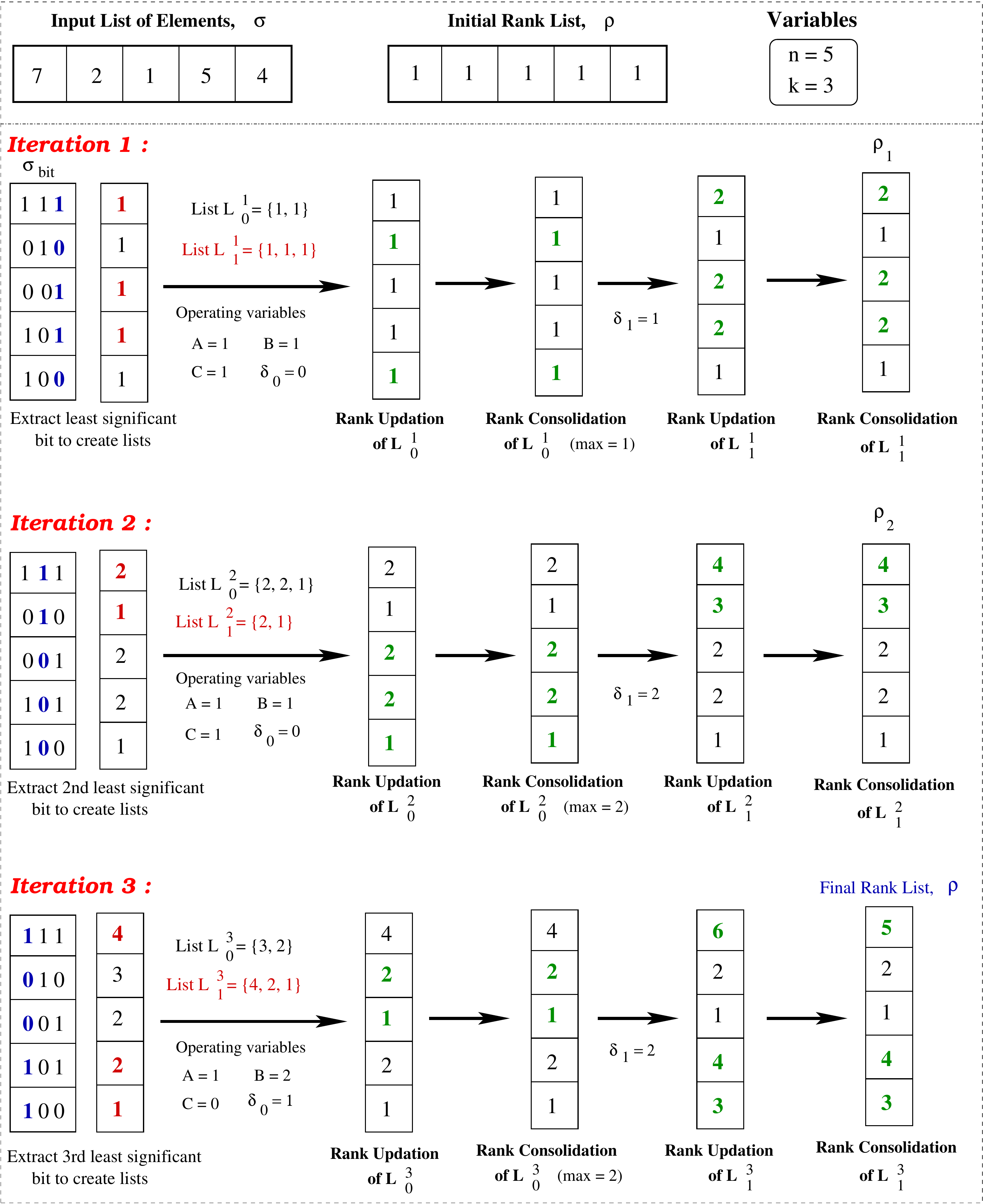}
	\caption{Pictorial depiction of the working of $SONIK$ on a running example.}
	\label{fig:example}
\end{figure}

\noindent \textbf{Working Example:} Continuing our previous example from Section~\ref{sec:sonik}, we now briefly walk through the operations 
of $SONIK$ as depicted in Figure~\ref{fig:example}. For the input element list $\sigma =\{7,2,1,5,4\}$, during the first iteration, list 
$L_1^1$ contains the current ranks of elements $\{7, 1, 5\}$ (having LSB set to $1$), while $L_0^1$ contains the ranks of the remaining 
items. Computing the operating variables, we obtain the update parameter $\delta_0 = 1$ for rank updation for list $L_0^1$. Since the updated 
ranks returned are already contiguous, no change in the element ranks is perceived during the rank consolidation phase, which returns $max=1$ 
representing the maximum rank assigned during this step. The minimum rank of list $L_1^1$ should now be set to $max+1$, hence the update 
parameter $\delta_1 = 1$ is used for re-assigning the element ranks in list $L_1^1$. Finally, after consolidation of $L_1^i$, the obtained rank 
list at the end of the first iteration, $\rho_1$ forms the input for the next iterations, wherein a similar procedure is followed for the 
remaining bits of the input elements.

Finally observe that, during the last iteration, the rank updation of $L_1^3$ generates a rank list of $\{6,4,3\}$, for the elements $\{7,5,4\}$, 
based on the update parameter $\delta_1 = 2$. The rank consolidation phase is now seen to provide a contiguous rank list, removing ranking ``holes'' 
in $L_1^1$ (i.e., the absence of rank $5$), to generate a consolidated rank list as $\{5, 4, 3\}$.

\noindent \textbf{Discussion:} $SONIK$ thus provides an efficient linear time and space approach (discussed in Section~\ref{ssec:anly}) for computing 
the ranks (based on the non-decreasing sorted order) of all the input elements. Observe that our approach does not involve any item re-ordering 
(which might be restricted in distributed settings) or dependence on other sorting techniques.

Interestingly, $SONIK$ depicts an {\em anytime algorithm}, wherein the rank computations can be terminated at the end of any iteration $i$, and the 
current rank list $\rho_i$ provides the ranking of elements based on only the $i$ least significant bits. This characteristic enables $SONIK$ to 
effectively cater to hard real-time approximate ranking application scenarios.

Further observe that, at each iteration, the working of our algorithm involves the processing of only the $i^{th}$ bits of all the input elements, 
without any other dependencies. As such, {\em on-demand} extraction of the appropriate bits incurs low network communication overhead (for items 
stored in distributed environments) at any particular time frame.

Finally, the output rank list for the input elements can be utilized to generate a sorted order of the input items. This can be performed in a 
single pass over the input and the rank lists (i.e., $O(n)$ time), by placing each element in the bucket with number corresponding to its computed 
rank, thus providing a linear time and space sorting algorithm.

\subsection{Performance Analysis}
\label{ssec:anly}

In this section, we formally analyze the performance of the proposed $SONIK$ algorithm for the {\em all item rank generation} problem. Assume the 
input element list $\sigma$ to be composed of $n$ integers, with each element represented by $k$ bits.

From Algorithm~\ref{algo:sonik}, we observe that at iteration $i$, the $i^{th}$ least significant bit of the input elements are extracted (ll. 5) 
and correspondingly the lists $L_0^i$ and $L_1^i$ are constructed, which can be performed in $O(n)$ time with a single pass over the input 
list, $\sigma$. Each element in $\sigma$ contributes an entry in either of the two lists (ll. 7--11) and hence the combined size of the obtained 
lists is $| L_0^i \cup L_1^i | = |\sigma| = n$.

The two phases of $SONIK$ -- rank updation and rank consolidation -- are invoked separately on the two lists (ll. 13--16), and hence each of the 
procedures operates on a total of $n$ elements. The rank updation process iterates over each of the elements of the lists and appropriately 
updates the rank using the {\em updation parameters}, $\delta_0$ and $\delta_1$. Thus, the time taken by the {\em rank updation} procedure 
is bounded by the total size of the two lists, i.e., $O(n)$. 

The rank consolidation procedure, for each of the lists, creates an additional storage $E$ of size $2*n$ and consolidates (or removes ``holes'' from) 
the input element ranks by a single pass over $E$, with a time complexity of $O(n)$ as input rank is pre-sorted (as described earlier in 
Section~\ref{ssec:l0}). The computation of the operating 
variables can also be computed in linear time by two iterations over the lists $L_0^i$ and $L_1^i$. Hence, the total time taken at each iteration $i$ 
is bounded by $O(n)$ (compared to $O(n \log n)$ for state-of-the-art approaches), considering assignments and simple mathematical computations to be 
constant time operations.

Since, the total number of iterations is $k$ (number of bits for input elements), the total running time complexity of $SONIK$ thus becomes $O(nk)$. 
Considering $k$ to a small constant (in the order of few tens) in most practical scenarios, the run-time of $SONIK$ can be considered to be 
{\em linear} in the size of the input list (similar to the performance analysis of {\em radix-sort}) for which the element ranks are to be computed.

Also, the space complexity of $SONIK$ is linear in the number of input elements, as the space required by the construction of lists and the extra 
storage during rank consolidation is bounded by $O(n)$. Hence, our proposed algorithm provides a time and space efficient real-time approach for 
computing the ranks of input elements without element re-ordering, use of sorting methods, or other input characteristic assumptions.

\section{Generalizations and Optimizations}
\label{sec:opt}

In this section, we discuss the generalizability of $SONIK$ for handling different input element characteristics (such as negative values), and 
also possible avenues for optimizing the practical performance.

\noindent \textbf{(1) Negative Elements:} We initially assumed the input element list to be composed of non-negative numbers only. Modern computers 
store negative numbers using the {\em two's complement} binary representation technique~\cite{two}. In this approach, the most significant bit is set 
to $1$ to denote a negative value. As such, during the last iteration step of $SONIK$ (i.e., the $k^{th}$ iteration), the negative elements will 
have their bits set and will mistakenly be considered to have a large numerical value and thus assigned larger ranks. To address such input scenarios, 
we initially perform an $O(n)$ pre-processing step, wherein the largest negative number present in the input ($max_{neg}$) is identified and all the 
input elements are then incremented by $|max_{neg}|$ value. This ensures that the input list provided to $SONIK$ strictly contains non-negative elements. 

\noindent \textbf{(2) Duplicate Elements:} In general, the working of $SONIK$ inherently handles the presence of duplicate elements in the input list, 
and appropriately assigns the same rank to duplicate items. For example, the final element rank list $\rho = \{1,2,1\}$ will be generated corresponding 
to the input list $\sigma = \{2, 4, 2\}$. However, enforcing the {\em stability} criteria of sorting for obtaining the element ranks, the rank list 
for the above example should be $\{1,2,3\}$. To this extent, we use a variant of the counting sort procedure on the obtained element ranks as a 
post-processing step. The counting sort procedure operates on the generated rank list to provide a re-ordering of the ranks compliant with the 
stability property. Observe that, this approach requires $O(n)$ time and also $O(n)$ space as the final element ranks obtained range from $1$ to $n$. 

\noindent \textbf{(3) Optimizing $k$:} $SONIK$ utilized the weighted bit representation of elements to compare the numerical values of the input 
elements and generate the final rankings. As such, a mapping of the original input items to a different set of elements preserving the sorted 
ordering among the initial elements, and the rank computation thereon also produces a correct ranking list. We use this intuition for possible 
reduction in the number of iteration steps, 
thereby improving the practical run-time performance of our algorithm. The input elements are pre-processed to identify the minimum item, $min$, 
and all the input numbers are subsequently decremented by a value of $min$. Observe that, this translation of the input items preserves the sorted 
order among the elements, and it is on these ``new'' items that the element ranks are computed. This process might help in possibly reducing $k$, 
the number of bits required to represent the input elements (as the maximum value among the input elements is reduced), thereby enhancing the 
practical performance of $SONIK$.

\noindent \textbf{(4) Sorting Order:} In scenarios where the required element rankings are to be based on the non-increasing sorted order (previously 
we considered non-increasing order), the operational stages and working of $SONIK$ remain nearly unchanged, thus providing a robust framework. In such 
cases, only the generation of the lists $L_0^i$ and $L_1^i$ are reversed, that is, at iteration $i$, elements with the $i^{th}$ bit set to $1$ are 
considered for constructing the $L_0^i$ list, while the other elements are used for list $L_1^i$.

\section{Conclusions}
\label{sec:conc}

This paper presented $SONIK$, a {\em linear time and space} algorithm for computing the rankings (based on sorted order) for all elements of an input 
list. Our algorithm uses bit operations to compare the numerical values of elements for assigning appropriate ranks, based on two operational steps -- 
rank updation and rank consolidation. The ranks are efficiently generated in the order of element appearance in the input list, and hence eliminates the 
need for item re-ordering or intermediate sorting as compared to state-of-the-art approaches. $SONIK$ is also shown to demonstrate robust characteristics 
like {\em on-demand bit extraction} and {\em anytime algorithm}, also providing a linear-time sorting technique. Extension of our method to handle strings 
and other input data types provides an interesting area of future studies.

\bibliographystyle{plain}
\bibliography{rank}

\end{document}